\documentclass[sigconf]{acmart}

\AtBeginDocument{%
  \providecommand\BibTeX{{%
    \normalfont B\kern-0.5em{\scshape i\kern-0.25em b}\kern-0.8em\TeX}}}

\copyrightyear{2023}
\acmYear{2023}
\setcopyright{acmlicensed}\acmConference[WWW '23]{Proceedings of the ACM Web Conference 2023}{April 30-May 4, 2023}{Austin, TX, USA}
\acmBooktitle{Proceedings of the ACM Web Conference 2023 (WWW '23), April 30-May 4, 2023, Austin, TX, USA}
\acmPrice{15.00}
\acmDOI{10.1145/3543507.3583450}
\acmISBN{978-1-4503-9416-1/23/04}

\usepackage{algorithm}
\usepackage{algorithmic}
\usepackage{url}
\usepackage{multirow}
\usepackage{subfigure}
\usepackage{enumitem}
\usepackage{color}

\renewcommand{\algorithmicrequire}{ \textbf{Input:}} 
\newcommand{\partitle}[1]{\medskip \noindent \textbf{#1.}}

\newcommand{\vect}[1]{\ensuremath{\mathbf{#1}}}
\newcommand{\mat}[1]{\ensuremath{\mathbf{#1}}}
\newcommand{\vh}{\vect{h}}
\newcommand{\vt}{\vect{t}}
\newcommand{\vr}{\vect{r}}
\newcommand{\mE}{\mat{E}}
\newcommand{\mG}{\mat{G}}
\newcommand{\mR}{\mat{R}}

\newcommand{\projectname}{\textit{DP-FLames }}
\newcommand{\projectnameadp}{\textit{DP-FLames-Adp }}

\begin{document}


\title{Quantifying and Defending against Privacy Threats on Federated Knowledge Graph Embedding}

\author{Yuke Hu}
\orcid{0000-0001-5780-6898}
\affiliation{
  \institution{Zhejiang University, HIC-ZJU}
  \city{Hangzhou}
  \country{China}
}
\email{yukehu@zju.edu.cn}

\author{Wei Liang}
\orcid{0000-0001-7115-5593}
\affiliation{
  \institution{Zhejiang University, HIC-ZJU}
  \city{Hangzhou}
  \country{China}
}
\email{liangwei23@zju.edu.cn}

\author{Ruofan Wu}
\orcid{0000-0002-2005-6058}
\affiliation{
  \institution{Ant Group}
  \city{Shanghai}
  \country{China}
}
\email{ruofan.wrf@antgroup.com}

\author{Kai Xiao}
\orcid{0000-0002-5181-1917}
\affiliation{
  \institution{Ant Group}
  \city{Shanghai}
  \country{China}
}
\email{xiaokai.xk@antgroup.com}

\author{Weiqiang Wang}
\orcid{0000-0002-6159-619X}
\affiliation{
  \institution{Ant Group}
  \city{Hangzhou}
  \country{China}
}
\email{weiqiang.wwq@antgroup.com}

\author{Xiaochen Li}
\orcid{0000-0002-3722-4783}
\affiliation{
  \institution{Zhejiang University, HIC-ZJU}
  \city{Hangzhou}
  \country{China}
}
\email{xiaochenli@zju.edu.cn}
 
\author{Jinfei Liu}
\orcid{0000-0003-2921-2827}
\affiliation{
  \institution{Zhejiang University, HIC-ZJU}
  \city{Hangzhou}
  \country{China}
}
\email{jinfeiliu@zju.edu.cn}

\author{Zhan Qin}
\orcid{0000-0001-7872-6969}
\affiliation{%
  \institution{Zhejiang University, HIC-ZJU}
  \city{Hangzhou}
  \country{China}
}
\email{qinzhan@zju.edu.cn}
\authornote{Corresponding author}
	




\renewcommand{\shortauthors}{Hu and Liang, et al.}

\begin{abstract}
    Knowledge Graph Embedding (KGE) is 
    a fundamental technique that extracts expressive representation from knowledge graph (KG) to facilitate diverse downstream tasks.
    The emerging federated KGE (FKGE) collaboratively trains from distributed KGs held among clients while avoiding exchanging clients' sensitive raw KGs,
    which can still suffer from privacy threats as evidenced in other federated model trainings (e.g., neural networks). However, quantifying and defending against such privacy threats remain unexplored for FKGE which possesses unique properties not shared by previously studied models.
    In this paper, we conduct the first holistic study of the privacy threat on FKGE from both attack and defense perspectives.
    For the attack, we quantify the privacy threat by proposing three new inference attacks, which reveal substantial privacy risk by successfully inferring the existence of the KG triple from victim clients.
    For the defense, we propose DP-Flames, a novel differentially private FKGE with private selection, which offers a better privacy-utility tradeoff by exploiting the entity-binding sparse gradient property of FKGE and comes with a tight privacy accountant by incorporating the state-of-the-art private selection technique.
    We further propose an adaptive privacy budget allocation policy to dynamically adjust defense magnitude across the training procedure.
    Comprehensive evaluations demonstrate that the proposed defense can successfully mitigate the privacy threat by effectively reducing the success rate of inference attacks from $83.1\%$ to $59.4\%$ on average with only a modest utility decrease.
    
\end{abstract}





\maketitle
\section{Introduction} \label{sec: introduction}
Recent years have witnessed substantial development in knowledge graphs (KGs).
Many large-scale KGs, e.g., Freebase \cite{DBLP:conf/sigmod/BollackerEPST08}, NELL \cite{DBLP:conf/aaai/CarlsonBKSHM10}, YAGO \cite{DBLP:conf/nips/SocherCMN13}, and Wikidada \cite{DBLP:journals/cacm/VrandecicK14} are constructed and applied to a variety of downstream applications such as recommendation systems \cite{DBLP:conf/kdd/ZhangYLXM16}, semantic web \cite{DBLP:conf/semweb/WangWLCZQ18}, knowledge reasoning \cite{DBLP:conf/nips/BordesUGWY13}, and question answering \cite{DBLP:conf/emnlp/BordesCW14}.
Knowledge graph embedding (KGE) is proposed to tackle the underlying symbolic nature of KG that hinders a more widespread application by representing the components of KG as dense vectors in continuous space (also known as embedding space) \cite{DBLP:conf/nips/BordesUGWY13}.
Most KGs owned by various organizations (e.g., e-commerce companies and financial institutions) are incomplete and cannot be exchanged directly due to privacy concerns and regulatory pressure (e.g., GDPR \cite{GDPR} and CCPA \cite{CCPA}).
In order to exploit the complementarity between different KGs without compromising privacy, several federated KGE (FKGE) models \cite{DBLP:conf/jist/ChenZYJC21, DBLP:journals/corr/abs-2203-09553, DBLP:conf/cikm/PengLSZ021} are proposed to collaboratively train a global KGE model while preserving the sensitive KGs local to the clients by exchanging only model parameters (i.e., embedding matrices in FKGE), as illustrated in Figure \ref{fig: fede}.
\begin{figure}[ht]
	\centering
	\includegraphics[width=0.40\textwidth]{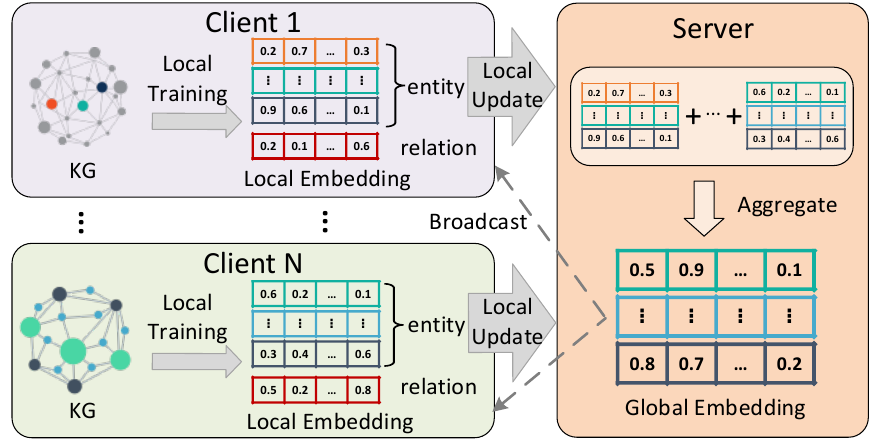}
	\vspace{-2.0ex}
	\setlength{\belowcaptionskip}{-2ex}
	\caption{Federated Knowledge Graph Embedding Training.}
	\label{fig: fede}
        \Description[A framework of federated knowledge graph embedding training]{Several clients collaboratively train a global KGE model while preserving the sensitive KGs local by exchanging only model parameters.}
	\vspace{-1.1em}
\end{figure}

However, it becomes a growing consensus that private information can still be divulged by the exchanged model parameters, as witnessed in other federated model trainings (e.g., federated neural network trainings) \cite{DBLP:conf/sp/NasrSH19, DBLP:conf/sp/MelisSCS19, DBLP:conf/nips/GeipingBD020}. 
Please refer to Appendix \ref{Appendix: related work} for more related work.
Despite the increasing deployment of FKGE across diverse areas involving sensitivity KGs, its privacy vulnerability still remains unexplored, which calls for a thorough investigation into the critical questions: Q1. To what degree the privacy of the clients' sensitive KGs can be inferred by privacy attacks? Q2. How to defend against such attacks with rigorous privacy guarantees?

In this paper, we therefore aim to fulfill this gap by conducting the first holistic study for the privacy threats on FKGE, which will significantly deepen our understanding from both privacy attack and defense perspectives. The proposed attack and defense mechanisms are tailored to the unique properties of FKGE, thus distinct from those designed for the previously explored models like neural networks in literature. 

\partitle{Attack}
We answer question Q1 by quantifying the privacy risks of FKGE through the lens of membership inference attacks (MIAs). In particular, we propose three new MIAs with the attack purpose of inferring whether a target KG triple exists in the training KG of the victim client or not, including server-initiate inference attack, client-initiate passive inference attack, and client-initiate active inference attack. To comprehensively inspect the privacy vulnerability of FKGE, these attacks each target different attack settings, where the adversary varies in 1) identity: server or clients; 2) capability: passively following the standard FKGE training or actively deviating from the standard to gain greater attack capability; 3) auxiliary dataset: the availability of an auxiliary KG dataset or not. Our attacks can successfully infer the target KG triple from the victim client with a high F1-score, which reveals that FKGE can suffer from significant privacy leakage without further defense. In addition, we launch attacks at different rounds across the FKGE training procedure and find that the attacks manifest dynamic attack strengths. Such phenomenon suggests that the defense mechanisms should adaptively vary in defense magnitude to achieve a better tradeoff between defense and utility. 

\partitle{Defense} 
To answer question Q2, we mitigate the privacy risks of FKGE by proposing the first rigorous differentially private \cite{DBLP:conf/icalp/Dwork06} protection for FKGE. Differential privacy \cite{DBLP:conf/icalp/Dwork06} is a rigorous mathematical notion to defend against adversaries with arbitrary prior knowledge. DPSGD \cite{DBLP:conf/ccs/AbadiCGMMT016} is a candidate technique for the FKGE application and we consider it as a baseline approach by straightforwardly applying DPSGD to FKGE. However, the baseline approach fails to maintain a decent privacy-utility tradeoff, which has the root causes in both the large model size and the failure to make use of the entity-binding sparse gradient property of FKGE. When applied to neural networks, DPSGD is known to provide acceptable accuracy only for small-scale models \cite{DBLP:conf/sp/Yu0PGT19, DBLP:journals/corr/abs-1908-07643, DBLP:conf/ccs/AbadiCGMMT016} but struggles to maintain decent utility for large-scale models
. Unfortunately, KGE models are usually at very large scale. For example, the TransE model \cite{DBLP:conf/nips/BordesUGWY13} has about two million model parameters on the Fb15k-237 dataset \cite{DBLP:conf/emnlp/ToutanovaCPPCG15} with an embedding dimension of 128.  
Albeit large in model parameter size, the ``active size'' of the stochastic gradient is usually much smaller for FKGE due to the entity-binding sparse gradient property, which is specific to FKGE but not possessed by general neural network models. 

Motivated by the two limitations above, we propose \textit{DP-FLames}, \underline{\textbf{D}}ifferentially \underline{\textbf{P}}rivate \underline{\textbf{F}}ederated know\underline{\textbf{L}}edge gr\underline{\textbf{a}}ph e\underline{\textbf{m}}bedding with privat\underline{\textbf{e}} \underline{\textbf{s}}election, to exploit the entity-binding sparse gradient property of FKGE to achieve a better privacy-utility tradeoff. In particular, we design the private active gradient elements selection mechanism to avoid perturbing the inactive gradient elements that incur an overwhelming amount of noise, which keeps the ``active size'' relatively small. Meanwhile, we ensure the selection itself satisfies differential privacy and account its privacy loss in a tight manner to avoid the extra noise paid to privacy selection exceeding the noise savings from the reduced gradient size. Further, inspired by the varying attack strengths across the FKGE training procedure, we design an adaptive privacy budget allocation strategy to adjust the defense magnitude accordingly. Finally, although differential privacy has been investigated for KGE privacy protection in both centralized \cite{DBLP:journals/ws/HanDGCB22} and federated settings \cite{DBLP:conf/cikm/PengLSZ021}, existing attempts are not accurate enough with privacy analysis, which renders their defense mechanism not rigorous in differential privacy guarantee. Please refer to Appendix \ref{appendix: inaccurate DP} for further discussion. 
Also, \projectname has very distinct DP mechanism designs and more involved privacy analysis compared to existing works. 

In summary, we make the following contributions:
\begin{itemize}[leftmargin=*]
    \item We conduct the first holistic study for the privacy threat on federated knowledge graph embedding, which covers both the attack and defense aspects.
    \item We propose three new inference attacks to comprehensively quantify the privacy risks of FKGE, which target diverse attack settings in terms of the adversary's identity, capability, and auxiliary information.
    \item We propose the first rigorous differentially private FKDE mechanism, which offers a better privacy-utility tradeoff by exploiting the entity-binding sparse gradient property of FKGE training and comes with a tighter privacy loss accountant by introducing the state-of-the-art private selection technique.
    \item We conduct extensive experimental evaluations on three real-world KG datasets and four FKGE models to demonstrate that the proposed attacks can pose significant privacy threats on FKGE and verify that \projectname can effectively protect FKGE from inference attacks while maintaining decent utility.
\end{itemize}
\section{preliminaries}

\subsection{Federated Knowledge Graph Embedding}\label{section: Federated Knowledge Graph Embedding}
We begin by introducing the formal definitions of the knowledge graph and the knowledge graph embedding model, as follows,
\begin{definition} \label{def: knowledge graph} (\textbf{Knowledge Graph $\mathcal{G}$ and its Embedding)}.
The knowledge graph $\mathcal{G}$ is defined as $\mathcal{G}=\{\mathcal{E}, \mathcal{R}, \mathcal{T}\}$, where $\mathcal{E}, \mathcal{R}, \mathcal{T}$ are the sets of entities, relations, and triples, respectively. A triple represents a fact $(h,r,t)\in \mathcal{E}\times \mathcal{R} \times \mathcal{E}$, which means a head entity $h$ and a tail entity $t$ is connected by a relation $r$. The KGE model represents the entities $h,t$ and relation $r$ as dense vectors $\vh, \vt, \vr$ in embedding space, which has the following per-triple loss function,
\begin{equation*} 
    \mathcal{L}(\vh,\vr,\vt)=-\log\sigma(f_r(\vh,\vt)-\gamma)-\sum_{i=1}^n p_i\log\sigma(\gamma-f_r(\vh,\vt'_i))),
\end{equation*}
where $f_r(\vh,\vt)$ is the score function that measures the plausibility of $(\vh,\vr,\vt)$, $\gamma$ is the margin, $(\vh,\vr,\vt') \notin \mathcal{T}$ is a negative triple generated by replacing the original tail entity with a random entity, n is the number of negative triples, and $p_i$ is the weight.
\end{definition}
Common KGE models include TransE \cite{DBLP:conf/nips/BordesUGWY13}, RotateE \cite{DBLP:conf/iclr/SunDNT19}, DisMult \cite{DBLP:journals/corr/YangYHGD14a}, and ComplEx \cite{DBLP:conf/icml/TrouillonWRGB16},
which have different forms of score functions. 
A triple with a higher score is more likely to be valid.
Score functions $f_r(h,t)$ of these models are described in Appendix \ref{appendix: score function}.

Federated KGE learning can further improve embeddings of all participating KGs by exploiting the overlapping entities among them, i.e., the entities in one KG usually overlap with the entities in other KGs \cite{DBLP:conf/cikm/PengLSZ021, DBLP:conf/jist/ChenZYJC21}.
The federated KGE is formally defined as follows.
\begin{definition} \label{def: federated KGE} (\textbf{Federated Knowledge Graph Embedding})
There are $m$ knowledge graphs $\{\mathcal{G}_i\}_{i=1}^m$ with potentially overlapping entity sets, each held by one client.
For the $\kappa_{th}$ round of FKGE training, the $i$-th client runs local updating steps on its local KG $\mathcal{G}_i$, local relation embedding ${\mR}_{\kappa-1}^i$ and local entity embedding  ${\mE}_{\kappa-1}^i$ for certain iterations, and then sends the updated local embedding ${\mE}_{\kappa}^i$ to the server;
The server receives $\{{\mE}_{\kappa}^i\}_{i=1}^m$ from $i\in \{1,...,m\}$ for aggregation and broadcasts the global embedding ${\mE}_{\kappa}$ to all clients.
Such communication rounds repeat until convergence.
\end{definition}

\subsection{Differential Privacy}
As a mathematical framework against adversaries with arbitrary prior knowledge, differential privacy \cite{DBLP:conf/icalp/Dwork06} can quantitatively analyze and bound the privacy loss of a data-dependent algorithm.
\begin{definition}\label{def: differential privacy} (\textbf{Differential Privacy} \cite{DBLP:conf/icalp/Dwork06})
A randomized algorithm $\mathcal{M}$ satisfies $(\epsilon, \delta)$-DP if for any two neighboring databases $D$ and $D'$ that differ in only a single entry, and for all possible outcome sets, it holds that $\mathcal{O}\in Range(\mathcal{M})$: $\Pr[\mathcal{M}(D)\in \mathcal{O}] \leq e^\epsilon\Pr[\mathcal{M}(D')\in \mathcal{O}] + \delta.$
\end{definition}
Differential privacy ensures that an adversary could not reliably infer whether one particular sample is in the dataset or not, even with arbitrary side information. 
Mironov introduces the notion of R\'{e}nyi differential privacy (RDP) \cite{DBLP:conf/csfw/Mironov17} based on the $\alpha$-R\'{e}nyi divergence $\mathbb{D}_\alpha(\cdot\|\cdot)$, which facilitates tighter privacy loss accountant. We introduce the approximate RDP definition below, which can be seen as a generalization of the original RDP introduced in \cite{DBLP:conf/csfw/Mironov17}).
\begin{definition}\label{def: approximated renyi differential privacy} (\textbf{Approximated R\'{e}nyi Differential Privacy} \cite{DBLP:conf/aistats/0005W22}).
A randomized algorithm $\mathcal{M}$ satisfies $\delta$-approximate-$(\alpha, \epsilon(\alpha))$-RDP if for any two neighboring database $D$ and $D'$, there exist events $E$ (depending on  $\mathcal{M}(D)$) and $E'$ (depending on  $\mathcal{M}(D')$) such that $\Pr[E]>1-\delta \wedge \Pr[E']>1-\delta$, and $\mathbb{D}_\alpha(\mathcal{M}(D)|E\|\mathcal{M}(D')|E') \le \epsilon(\alpha)$ for any $\alpha \ge 1$, where
\begin{small}
\begin{equation*}
    \mathbb{D}_\alpha(\mathcal{M}(D)\|\mathcal{M}(D')) := \frac{1}{\alpha-1}\log\mathbb{E}_{o\sim \mathcal{M}(D)}[(\frac{\Pr[\mathcal{M}(D)=o]}{\Pr[\mathcal{M}(D')=o]})^\alpha].
\end{equation*}
\end{small}
\end{definition}
The above $\delta$-approximate-$(\alpha, \epsilon)$-RDP reduces to the original $(\alpha, \epsilon)$-RDP  when $\delta=0$. Other properties are shown in Appendix \ref{appendix: RDP}.

\section{Attack} \label{sec: attack}
In order to quantify the privacy threat on vanilla FKGE, we propose three attacks to infer the existence of particular KG triples held by the victim clients. We first describe the attack settings and then present the detailed mechanisms of the new attacks.
The important notations are shown in Table \ref{table: notation} at Appendix \ref{appendix: notation table} to avoid confusion.
\subsection{Attack Setting} \label{subsection: threat model}
\partitle{Adversary's Goal} 
The goal of the FKGE triple inference attack $\mathcal{A}$ is to determine whether a target KG triple $(h,r,t)$ is in the training set of the victim client in a given FKGE learning, i.e., the adversary has the FKGE triple inference attack $\mathcal{A}(h,r,t)$ to output whether $(h,r,t)$ \emph{exist} or \emph{non-exist} in the KG training set of the client.
Since triples contain sensitive information of the victim client, the success of $\mathcal{A}(h,r,t)$ will incur severe privacy threats on the victim client.

\partitle{Adversary's Background Information and Capability}
To comprehensively inspect all potential privacy threat initiators in FKGE, we consider both the server and clients to be the possible adversary and delineate their background knowledge and capabilities according to their roles played in FKGE.
\begin{itemize}[leftmargin=*]
    \item Server as Adversary: the adversary has the entire entity set $\mathcal{E}$ (including the victim client's $\mathcal{E}_v$) and the periodically uploaded entity embedding matrix $\mE$ from all clients (including the victim client's $\mE_v$). In addition, to understand the upper limit of the inference capability under such setting, we provide the adversary with as ample auxiliary information as possible. Thus, the adversary has access to an auxiliary KG dataset $\mathcal{D}_{aux}$ from the same knowledge domain with the FKGE learning. In practice, such auxiliary KG dataset can come from publicly available sources  (e.g., Wikipedia) or be built upon empirical common sense (e.g., the relation between \textit{patient} and a type of \textit{disease} is \textit{diagnosed with}). We assume the server broadcasts normal global embeddings without modifying them.  
    \item Client as Adversary: the adversary has its local KG dataset and all the global entity embedding matrix broadcast by the server at each communication round. The malicious client can either follow standard FKGE protocol to send normally updated embeddings to the server (i.e., passive attack) or alter the embeddings in order to better infer the victim clients' triples (i.e., active attack). 
\end{itemize}


\begin{table}[htbp] \small
    \vspace{-1em}
    \caption{An overview of three proposed attacks.}
    \vspace{-1.5em}
    \label{table: attack}
    \begin{tabular}{cccc}
      \toprule
      Attack type & Identity & Capability & Auxiliary dataset\\
      \midrule
      $\mathcal{A}_{SI}$ & server & passive & \checkmark\\
      $\mathcal{A}_{CIP}$ & client & passive & -\\
      $\mathcal{A}_{CIA}$ & client & active & -\\
      \bottomrule
    \end{tabular}
     \Description[An overview of three proposed attacks]{Server-initiate Inference Attack is passive and need auxiliary dataset. Client-initiate Passive Inference Attack is passive and doesn't need auxiliary dataset. Client-initiate Active Inference Attack is active and doesn't need auxiliary dataset.}
    \vspace{-1.5em}
\end{table}

\partitle{Attack Taxonomy}
Table \ref{table: attack} summarizes three attacks according to varying attack settings in terms of identities, capabilities, and auxiliary information, which are referred to as server-initiate inference attack ($\mathcal{A}_{SI}$), client-initiate passive inference attack ($\mathcal{A}_{CIP}$), and client-initiate active inference attack ($\mathcal{A}_{CIA}$).

\subsection{Server-initiate Inference Attack} 
In the server-initiate inference attack, given the auxiliary dataset, the adversary suffices to infer the existence of the relation between the targeted head and tail entities, then the triple (i.e., the particular type of the already inferred relation) can be inferred by querying the auxiliary dataset. The attack proceeds in three steps.

\partitle{Step 1: Calculate Entity Embedding}
During the attack rounds, the malicious server receives the full entity embedding matrix $\mE_v$ from the victim client.
It enumerates all potential relation embeddings according to the embedding relations designated by the KGE model. Taking TransE as an example, it calculates all potential $\vr$ by  $\vr = \vt - \vh$. More generally, 
the calculation is as follows,
\begin{equation} \nonumber 
    \begin{split}     
    {\bf{R}}_{SI} = \left\{{\bf{r}}_i\right\}_{i=1}^{n_v^2-n_v}, {\bf{r}} = f_{e-r}({\bf{e}}_j, {\bf{e}}_k), 0 \leq j, k < n_v,      
    \end{split} 
\end{equation} 
where $n_v$ is the number of entities in the victim client and $f_{e-r}$ is the representation of entities and relations. ${\bf{e}}_j, {\bf{e}}_k$ are entity embeddings in victim client's $\mE_v$. When $f_{e-r}$ is too complex for the $\vr$ conversion, we use  $\vr$ by  $\vr = \vt - \vh$ as an approximation for simplicity.
 
\partitle{Step 2: Relation Inference} 
Since not all enumerated relation embeddings correspond to a real existing relation, we then need to identify the ``real'' relation. Our key intuition is that the embeddings for the real relations tend to be more concentrated together in the embedding space, while the embeddings of the fake relations usually scatter around. 
Thus, we cluster the relation embeddings into $2n_r$ clusters and identify the highly concentrated cluster centers of the relation embeddings. Given a candidate triple $(h,r,t)$ to infer the existence, if the corresponding $f_{e-r}(\vh, \vt)$ belongs to one of the concentrated clusters, we regard there exists a relation $(h,x,t)$ with its specific relation type $x$ unknown.

\partitle{Step 3: Triple Inference} 
The adversary can infer the specific type of the relation $x$ with the help of auxiliary dataset $\mathcal{D}_{aux}$.
That is, the adversary searches for the auxiliary-triple $(\mathcal{E}_1, r_{aux}, \mathcal{E}_2)$ with $h \in \mathcal{E}_1$ and $t \in \mathcal{E}_2$ to find a relation type to match $x$.
Once the adversary matches $r_{aux}$ with $x$, i.e,. $x=r_{aux}$, it successfully infers the target triple $(h,r,t)$.

\subsection{Client-initiate Passive Inference Attack}
In client-initiate passive inference attack, the adversary is a ``curious but honest'' client, which infers the particular triple of the victim client but still follows the standard FKGE training procedure. 

\partitle{Step 1: Calculate Entity Embedding}
The adversary first compares its own uploaded embedding and global entity embedding broadcast back from the server to identify the overlapping entities. Because only the embeddings of the overlapping entities will be updated in the aggregation, which can be leveraged by the adversary to identify the overlapping entities with the victim client. Then, the adversary obtains the target entity' embedding by extracting from the global embedding according to $\mathcal{E}_{CIP}, \mE_{CIP} \gets (N \times \mE_b - \mE_u)/(N-1)$,
where $\mE_b$ is the global entity embedding broadcast back from the server, $\mE_u$ is the adversary's uploaded embedding and $N$ is the number of clients advertised by FKGE.

\partitle{Step 2: Calculate Approximate Relation Embedding}
%
Next, the adversary approximates the relation embedding of the target triple with its local relation embedding. The intuition is that the local relation embeddings will gradually contain more relation information of the victim knowledge graph with the progress of the FKGE training, which can be leveraged by the adversary to approximate the target embedding with increasing proximity. 

\partitle{Step 3: Triple Inference}
With the entity and relation embeddings of the target triple, the adversary infers the existence of the triple by thresholding on the value of the score function, as follows,
\begin{equation} \nonumber
    \mathcal{A}_{CIP}(h,r,t) = 
    \begin{cases}
    exist, &f_r(\bf{h_1},\bf{t_1}) / \it f_r(\bf{h_2},\bf{t_2}) \geq \tau_{\it CIP}, \\
    non-exist, &f_r(\bf{h_1},\bf{t_1}) / \it f_r(\bf{h_2},\bf{t_2}) < \tau_{\it CIP},
    \end{cases}
\end{equation}
where $\bf{h_1},\bf{t_1}$ are corresponding embeddings of $h, t$ in attack model's embeddings $\mE_{CIP}$, $\bf{h_2},\bf{t_2}$ are corresponding embeddings of $h, t$ in adversary's local embeddings $\mE_u$ and $\tau_{CIP}$ is tunable decision threshold. $\mathcal{A}(h,r,t)$ denotes the adversary's prediction on the sample $(h,r,t)$. The rationale is that the score function quantifies the plausibility of a triple, so a lower score indicates that the target triple is more likely to exist. 

\subsection{Client-initiate Active Inference Attack}
In the client-initiate active inference attack, the adversary is a malicious client that deviates from the standard training procedure by modifies its local embeddings to gain greater attack capability. 


\partitle{Step 1: Reverse Entity Embeddings}
During the attacking rounds, the adversary reverses the tail entity embedding of the target triple $(h,r,t)$ to deliberately increase the score to the value $s_1$. After several rounds of FKGE training, the score of $(h,r,t)$ will be reduced to $s_2$ if the triple indeed exists in the KG of the victim client. 

\partitle{Step 2: Triple Inference}
The adversary infers the existence of the target triple by thresholding on the score ratio of $s_1/s_2$
\begin{equation} \nonumber
    \mathcal{A}_{CIA}(h,r,t) = 
    \begin{cases}
    exist, &s_1 / s_2 \geq \tau_{CIA}, \\
    non-exist, &s_1 / s_2 < \tau_{CIA},
    \end{cases}
\end{equation}
where $\tau_{CIA}$ is tunable decision threshold.

We further discuss the limitations of the above attacks and what will happen when the server colludes with a client in Appendix \ref{appendix: attack limitations}.

\section{Defence Mechanism} \label{sec: defense mechanism}

\begin{figure*}[t]
	\centering
	\includegraphics[width=0.95\textwidth]{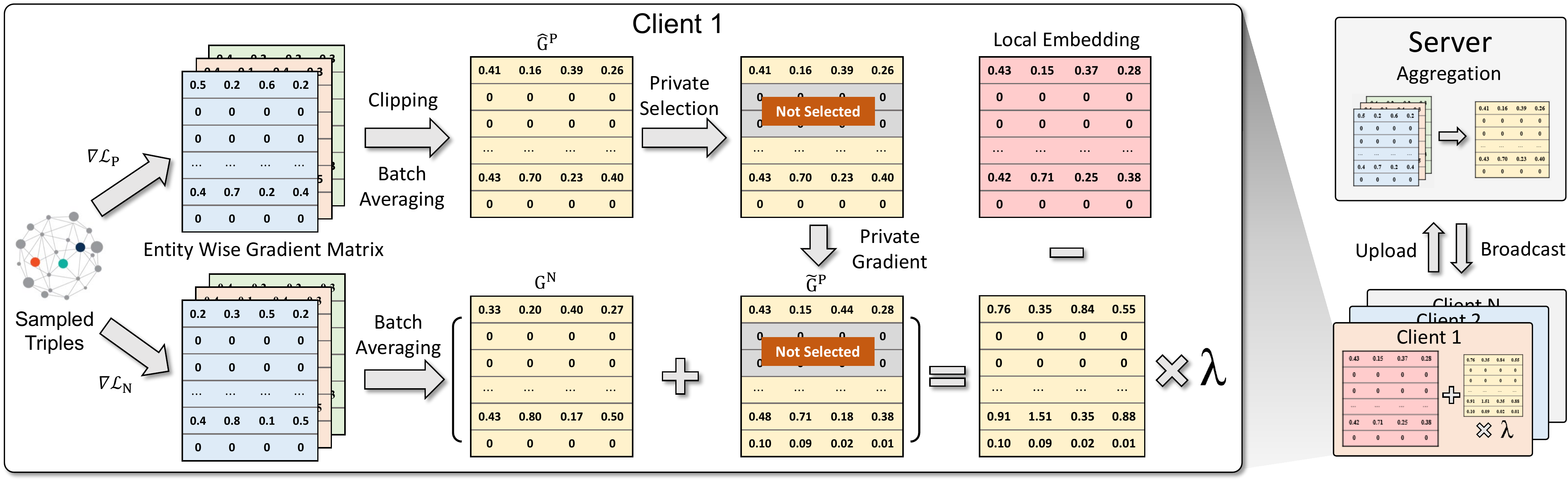}
	\setlength{\belowcaptionskip}{-2ex}
	\vspace{-0.25em}
	\caption{Overview of \projectname .}
	\label{fig: algorithm}
        \Description[The overview of DP-FLames framework]{1. The client calculate the entity wise gradient matrix. 2. Clip the gradient matrix and average it among the matrices. 3. Do private selection 4. Exchange parameters with the server.}
\end{figure*}

\begin{algorithm}[b]
	\caption{\projectname}
	\label{alg: dpsgd with private selection}
	{\small{
			\begin{algorithmic}[1]
				\REQUIRE {Triple set $\{(h,r,t)\}$, loss function $\mathcal L_{p/n}(\vh,\vr,\vt)$.
				Parameters: learning rate $\lambda$, batch size $B$, noise allocation parameters $\sigma, \eta, \Delta$, clipping bound $C_1, C_2$, private selection parameters $\epsilon_r$, $\sigma_p$, $\delta_t$}
				\COMMENT{ {\color{blue}The blue lines indicate advanced algorithm designs of \projectname.}}
				\renewcommand{\algorithmicrequire}{ \textbf{Server side procedure:}}
				\REQUIRE
				\STATE Align the entities for all clients
				\STATE Initialize and broadcast $\mE_0$ to all clients
				\FOR {round $\kappa\in [K]$}
				\STATE Collect $\mE_{\kappa}^c$ from each client $c$
				\STATE Compute $\mE_{\kappa+1}$ by averaging the embeddings of overlapped entities
				\STATE Broadcast $\mE_{\kappa+1}$ to all clients 
				\ENDFOR
                \renewcommand{\algorithmicrequire}{ \textbf{Client (c) side procedure:}}
				\REQUIRE
				\FOR {round $\kappa\in [K]$}
				\STATE Receive $\mE_\kappa$ from the server
				\FOR {$\tau\in [T]$}
				\STATE Take a random sample set $B$ with sampling probability $q = B/N$
				\STATE for each $i \in [B]$, compute $\mG_i^{p}\leftarrow \nabla \mathcal L_p(\vh_i,\vr_i,\vt_i)$ 
				\STATE $\widehat{\mG}_i^{p} \leftarrow \mG_i^{p}/\max\{1,\frac{\|\mG_i^{p}\|_2}{C_1}\}$ 
				\STATE {\color{blue} Compute $\overline{\mG}_i^{p}(j) \leftarrow \widehat{\mG}_i^{p}(j)/\max\{1,\frac{\|\widehat{\mG}_i^{p}(j)\|_2}{C_2}\}$ for rows in $\widehat{\mG}_i^{p}$}
				\STATE {\color{blue}$\mathcal{I}\leftarrow {\rm PrivSelection}(\sum_i\overline{\mG}_i^{p})$} 
				\STATE $\widetilde{\mG}^p(\mathcal{I}) \leftarrow \frac{1}{B}(\sum_i\overline{\mG}_i^{p}(\mathcal{I})+\mathcal{N}(0,\sigma^2C_1^2I_{k\times k}))$ 
				\STATE $\mG^{p}\leftarrow \frac{1}{N}\sum_{i}\mG_i^p$
				\STATE {\color{blue}$\mG^{n}\leftarrow \frac{1}{N}\sum_{i}\nabla \mathcal L_n(\vh'_i,\vr'_i,\vt'_i)$}
				\STATE $\mE^c_{\kappa_{\tau+1}} \leftarrow \mE^c_{\kappa_{\tau}}-\lambda (\widetilde{\mG}^p+\mG^n)$; $\mR^c_{\kappa_{\tau+1}} \leftarrow \mR^c_{\kappa_{\tau}}-\lambda (\mG^p+\mG^n)$ 
				\ENDFOR
				\STATE Upload $\mE^c_{\kappa+1} \leftarrow \mE^c_{\kappa_{T}}$ to server
				\IF {{\color{blue}$\kappa$ mod ${\rm validation\_interval}=0$ $\wedge$ ${\rm MRR}_t - {\rm MRR}_{t-1} < \Delta$}}
				\STATE {\color{blue} $\sigma = \eta \cdot \sigma$} \COMMENT{ {\color{blue}This line refers to the adaptive privacy budget allocation}}
				\ENDIF
				\ENDFOR
				\ENSURE {$\mE_K$, $\{\mE_K^c\}_{c=1}^m$}
	\end{algorithmic}}}
\end{algorithm}
\setlength{\floatsep}{1ex}

In this section, we propose \projectname to provide rigorous privacy protection for the FKGE model. We first present a baseline approach by straightforwardly applying DPSGD to FKGE to identify the following challenges: 1) how to avoid the significant decrease in model utility during privacy protection? 2) how to dynamically adjust the protection magnitude throughout the training process, given that the privacy attacks have varying attack strengths? \projectname is motivated by both challenges to simultaneously achieve rigorous differential privacy protection and decent model utility.

\subsection{Baseline Defense: DPSGD for FKGE}\label{subsection: baseline approach}
\partitle{Identify DP Granularity for KGE}
Since the two existing neighboring definitions for graph data \cite{DBLP:conf/icdm/HayLMJ09}, i.e., node-neighbor and edge-neighbor, are not suitable for the KGE model (as discussed in Appendix \ref{appendix: granularity}), 
we introduce the following KG triple-neighbor definition tailored to the DP protection for KGE.


\begin{definition}\label{def: neighboring knowledge graph}(\emph{Triple-level Neighboring Dataset for Knowledge Graph Embedding}).
Two knowledge graphs $\mathcal{G}_1=\{\mathcal{E}_1, \mathcal{R}_1, \mathcal{T}_1\}$ and $\mathcal{G}_2=\{\mathcal{E}_2, \mathcal{R}_2, \mathcal{T}_2\}$ are neighboring, if $\mathcal{E}_1=\mathcal{E}_2$, $\mathcal{R}_1=\mathcal{R}_2$, and $\exists(h,r,t) \in \mathcal{T}_1 \cup \mathcal{T}_2$, s.t. $(\mathcal{T}_1 \setminus \mathcal{T}_2)\cup(\mathcal{T}_2 \setminus \mathcal{T}_1)=\{(h,r,t)\}$.
\end{definition}

\partitle{Baseline Algorithm}
We demonstrate the baseline defense mechanism by applying DPSGD to KGE in the federated learning setting, which is based on the prominent DPSGD method in the differentially private deep learning area \cite{DBLP:conf/ccs/AbadiCGMMT016}. 
Algorithm \ref{alg: dpsgd with private selection} without {\color{blue}blue lines} provides summarization for the baseline algorithm.
It has three key steps on the client side: 
1) sample a random batch with sampling probability $q=B/N$ and compute per-example-gradient $\mG^i = \nabla \mathcal L(\vh, \vr, \vt)$ (Line 11-12 in Algorithm \ref{alg: dpsgd with private selection});
2) bound the influence of any sample by clipping the $\ell_2$ norm of the per-example gradient with a clipping threshold $C_1$, i.e. $\hat{\mG}^i = \mG^i/\max\{1,\frac{\|\mG^i\|_2}{C_1}\}$ (Line 13 in Algorithm \ref{alg: dpsgd with private selection});
3) inject Gaussian noise calibrated with a noise multiplier $\sigma$ (Line 16 in Algorithm \ref{alg: dpsgd with private selection}).
Then, the clients periodically upload the local updates to the server and receive the server aggregated global broadcasting.




\partitle{Limitation Analysis for the Baseline}
The limitation roots in both the large model size and the failure to make use of the entity-binding sparse gradient property of FKGE. 
KGE models are usually at very large scale. For example, TransE model on the Fb15k-237 dataset has two million model parameters.
Albeit large in model parameter size, the ``active size'' of the stochastic gradient is usually much smaller for FKGE models due to the entity-binding sparse gradient property.
The \emph{entity-binding sparse gradient property} originates from the loss function of the KGE model \cite{DBLP:conf/iclr/SunDNT19} $\mathcal{L}(\vh,\vr,\vt)$ (as shown in Section \ref{section: Federated Knowledge Graph Embedding}).
The per-triple gradient of entity embedding involves only $\vh$ and $\vt$ as its active elements, while all other elements will remain inactive. 
For the TransE model on the Fb15k-237 dataset with a batch size of 64, only $1\%$ are active gradient elements.
Intuitively, letting the $99\%$ zero-gradient elements remain zero in a differentially private way will avoid a great amount of unnecessary perturbation and improve the utility.

\subsection{Advanced Defense: \projectname} \label{subsection: advanced approach}

There are several techniques for screening out the active gradient such as sparse vector technique (SVT) \cite{DBLP:journals/pvldb/LyuSL17, DBLP:journals/fttcs/DworkR14, DBLP:conf/kdd/LeeC14} and $k$-time exponential mechanism \cite{DBLP:conf/nips/DurfeeR19}.
However, these techniques fail to provide tight privacy loss accountant. We utilize the state-of-the-art Report-Noisy-Max with propose-test-release (PTR) for data-dependent top-$k$ private selection \cite{DBLP:conf/aistats/0005W22} (as shown in Algorithm \ref{alg: private selection}) based on two reasons: (1) PTR privately selects $k$ vectors at one shot while SVT and exponential mechanisms need $k$-time composition of privacy budget; (2) although there is a probability that the test in PTR fails, which should be the main bottleneck of this solution, the privacy loss is under control since no extra privacy budget will be consumed by the private gradient step that will be discarded in this case. 
We describe the details of private selection (Algorithm \ref{alg: private selection}) as follows.


\partitle{Data-Dependent Choice of $k$} 
The first step is to determine the number of active gradient rows $k$.
Suppose that there are $n$ rows in the gradient matrix $\mG$, with each referring to the gradient vector of its entity's embedding. 
Only the gradient vectors correspond to the sampled batch are active. 
However, we can not preset a uniform $k$ because the number of involved entities varies from batch to batch.
Thus we adopt a data-dependent approach to determine $k$.
Specifically, we first sort the rows in $\mG$ by $\ell_2$ norm clipped with $C_2$ (Line 14 in Algorithm \ref{alg: dpsgd with private selection}) into descending order (Line 1) and then look for the maximal gap between any two adjacent vectors with Gumbel noise
\cite{DBLP:conf/nips/DurfeeR19} (Line 2).
The regularizer $r(j)$ in Line 2 encodes the prior knowledge that the actual $k$ should be within $[B, 2B]$, which takes the following form,
\begin{equation} \nonumber
    r(j) = 
    \begin{cases}
    c, &{\rm if} j\in[B,2B], \\
    -\infty, &{\rm if} j\notin [B,2B],\\
    \end{cases}
\end{equation}
where $c$ is an arbitrary constant that can be narrowed down given additional prior knowledge.
The maximal $\ell_2$-norm gap between $\|\mG(k)\|_2$ and $\|\mG(k+1)\|_2$ is recorded as $d_k$ for subsequent PTR process (Line 3).

\setlength{\floatsep}{1ex}
\setlength{\textfloatsep}{1ex}

\begin{algorithm}[htbp]
	\caption{Private Selection}
	\label{alg: private selection}
	{\small{
			\begin{algorithmic}[1]
				\REQUIRE {Gradient matrix $G$, clipping bound $C_2$, RDP parameters $\sigma_r$, $\sigma_{p}$, $\delta_t$}
				\STATE sort rows in $G$ by $\ell_2$ norm into a descending order: 
				$\|G(1)\|_2, \|G(2)\|_2, $ $\dots, \|G(n)\|_2$ 
				\STATE $k \leftarrow Argmax_{j\in[n-1}\{\|G(j)\|_2 - \|G(j+1)\|_2 + r(j) +\rm {Gumbel}(2C_2\sigma_r)$ 
				\STATE $d_k \leftarrow \|G(k)\|_2 - \|G(k+1)\|_2$ 
				\STATE $\hat d_k \leftarrow \max\{C_2, d_k\} + \mathcal{N}(0, \sigma_{p}^2C_2^2)-\sigma_{p}C_2\sqrt{2\log(1/\delta_t)}$ 
				\IF {$\hat d_k>C_2$}
				\STATE$\mathcal{I} \leftarrow \{i_{(1)}, \dots, i_{(k)}\}$
				\ELSE 
				\STATE $\mathcal{I}\leftarrow \perp$
				\ENDIF
				\ENSURE Index set of top-$k$ gradients $\mathcal{I}$
	\end{algorithmic}}}
\end{algorithm}
\setlength{\floatsep}{1ex}

\partitle{Active Gradient Indexes Release}
Once $k$ and $d_k$ are determined, we privately release the indexes of the top-$k$ gradient vectors.
Instead of calling SVT or exponential mechanisms for $k$ time, we use the PTR-based approach \cite{DBLP:conf/aistats/0005W22, DBLP:conf/stoc/DworkL09} to release $k$ indexes at one shot.
We first add Gaussian noise to the max of $d_k$ and $C_2$ as the ``proposed'' bound in the PTR framework (Line 4).
Then, we privately ``test'' whether $\hat d_k$ is a valid upper bound(Line 5).
The index set of the selected top-$k$ vector is "released" if $\hat d_k$ passes the test, otherwise nothing will be reported (Line 6-8).

\partitle{Negative Sampling} To take full advantage of the sparsity in the gradients, we modify the self-adversarial negative sampling. The original method only replaces the tail entity $t$ of a valid triple $(h,r,t)$ to a random $t'$ as the negative sample \cite{DBLP:conf/iclr/SunDNT19}.
The negative sample set usually consists of hundreds of negative triples, which makes the gradient much denser. 
We instead generate completely random negative sample $(h',r',t')$ instead of $(h, r, t')$ so that the negative gradient can be released without perturbation.
Whereas, the model utility is impaired when doing so because the randomly generated head-relation pair
is usually semantically meaningless.
There are two fixes: either by introducing real head-relation pairs from a public dataset, or by generating the negative samples from the private dataset with differentially private data synthesis techniques \cite{DBLP:conf/uss/00010LH0H0021, DBLP:journals/pvldb/GeMHI21, DBLP:conf/asiaccs/Li22}.
We adopt the first one for simplicity.
The loss function can be seperated into positive and negative loss components and the gradient of the positive loss component is still highly sparse.
\begin{small}
    \begin{gather}
        \mathcal{L}_{p}(\vh,\vr,\vt)=-\log\sigma(f_r(\vh,\vt)-\gamma),\\
        \mathcal{L}_{n}(\vh,\vr,\vt)= -\sum_{i=1}^n p(\vh,\vr,\vt')\log\sigma(\gamma-f_r(\vh,\vt'_i))).
    \end{gather}
\end{small}
The gradient for Eq.(1) is processed by \projectname, while the gradient for Eq.(2) is by vanilla SGD (Line 18 in Algorithm \ref{alg: dpsgd with private selection}).

\subsection{Adaptive Privacy Budget Allocation}\label{subsection: adaptive allocation}
In our experiment of active attack, the success rate peaks in the first few rounds of training and drops as the model converges. 
This motivates us to provide greater defense in the first few rounds and gradually lower the defense in the later rounds for better utility.

We propose to adaptively adjust the noise scale $\sigma$ based on the validation accuracy during the training procedure (Line 23-25 in Algorithm \ref{alg: dpsgd with private selection}).
Specifically, we set the initial $\sigma$ relatively large to defend against the inference attack and calculate the validation accuracy every few rounds.
If the verification accuracy no longer increases, we decrease $\sigma$ by multiplying it with a coefficient $\eta<1$ until convergence.
More details can be seen in Appendix \ref{appendix: adaptive allocation}.

\begin{table*}[hbt]
 \vspace{-1em}
  \caption{\small F1-score of Attacks.}
  \label{table: attack_f1score}
  \vspace{-4mm}
  \footnotesize
  \begin{center}
   \setlength{\tabcolsep}{1.7mm}{
    \begin{tabular*}{\textwidth}{ @{\extracolsep{\fill}} lccccccccccccc}
     \toprule
     \multirow{2.5}{*}{Dataset} &\multicolumn{4}{c}{\textbf{Client-Initiate Passive Inference}}  & \multicolumn{4}{c}{\textbf{Client-Initiate Active Inference}} & \multicolumn{2}{c}{\textbf{Server-Initiate Inference}} \\
     \cmidrule(lr){2-5} \cmidrule(lr){6-9} \cmidrule(lr){10-11} 
     & TransE & RotatE & DistMult & ComplEx & TransE & RotatE & DistMult & ComplEx & TransE & RotatE\\
     \midrule
     
     fb15k-237 & 0.724 & 0.8211 & 0.7809 & 0.7749 & 0.8779 & 0.8403 & 0.7346 & 0.7949 & 0.878 & 0.7827\\
     nell-955 & 0.8243 & 0.8728 &	0.8686 &	0.8429 & 0.8817 &	0.8462 & 0.7337	& 0.7848 & 0.9023 &	0.748\\
     \bottomrule
   \end{tabular*}}
   \end{center}
   \Description[F1-score of attacks]{CIP performs well on all four KGE algoritms with F1-Scrore around $0.8$. CIA usually outperforms CIP on TransE and RotatE with F1-Scrore around $0.85$, while does not perform as well as CIP on DisMult and ComplEx. SI only works on translational distance models, i.e., TransE and RotatE.}
   \vspace{-1em}
\end{table*}

\subsection{Privacy Analysis}
In this part, we provide the privacy analysis for \projectname. We first provide the per-iteration privacy analysis, which consists of the RDP privacy accountants for both the private selection part and the private gradient part. Then, we compose the per-iteration RDP privacy accountants over T iterations and convert the overall RDP result to DP guarantee. All proofs are deferred to Appendix \ref{appendix: proof}.

\partitle{Private Gradient Analysis}
The RDP accountant provides tight privacy loss tracking and has been adopted by many mainstream packages for DPSGD implementation, e.g., \cite{Tensorflow-Privacy, PyVacy, Pytorch-Opacus}. 
Lemma \ref{lemma:dpsgd} below shows the per-iteration privacy guarantee for the private gradient part in terms of RDP accountant, which includes Lines 11,13,16 in Algorithm \ref{alg: dpsgd with private selection}.
\begin{lemma}\label{lemma:dpsgd}
    The private gradient steps in Algorithm \ref{alg: dpsgd with private selection} satisfies $(\alpha, \epsilon_g(\alpha))$-RDP, where $\epsilon_g(\alpha)=\frac{2q^2\alpha}{\sigma}$ for
    $1 < \alpha \le \min(\frac{1}{2}\sigma^2L-2\log\sigma,  \frac{1/2\sigma^2L^2-\log5-2\log\sigma}{L+\log(q\alpha)+1/(2\sigma^2)}),$
    where the sampling ratio $q=B/N$ and $L=\log(1+\frac{1}{q(\alpha-1)})$.
\end{lemma}

\partitle{Private Selection Analysis}
The private selection steps consist of two stages: Report-Noisy-Max and Propose-Test-Release, and can be analyzed with approximated RDP. 
Theorem \ref{lemma:private selection} shows the per-iteration privacy guarantee for the private selection part in terms of RDP accountant with subsampling, which is shown in Algorithm \ref{alg: private selection} as well as Line 11 in Algorithm \ref{alg: dpsgd with private selection}.
\begin{theorem}\label{lemma:private selection}
    The private selection in Algorithm \ref{alg: private selection} with sampling ratio $q$ satisfies $\delta_t$-approximated-$(\alpha, \epsilon_s(\alpha))$-RDP, where
    \begin{small}
    \begin{align}
        \nonumber &\epsilon_s(\alpha)=\frac{1}{\alpha-1}\log(p^2\tbinom{\alpha}{2}\min\{4(e^{\epsilon(2)}-1), e^{\epsilon(2)}\min\{2, (e^{\epsilon(\infty)}-1)^2\}\}\\
        \nonumber &+\sum_{j=3}^\alpha q^j\tbinom{\alpha}{j}e^{(j-1)\epsilon(j)}\min\{2,(e^{\epsilon(\infty)}-1)^j\}+1)
        \text{and } \epsilon(\alpha)=\frac{\alpha}{8\sigma_r^2}+\frac{\alpha}{2\sigma_p^2}.
    \end{align}
    \end{small}
\end{theorem}

The client needs several iterations
in Algorithm \ref{alg: dpsgd with private selection} to train the local embedding. The privacy guarantee of $T$-iteration composition under RDP is shown in Theorem \ref{theorem: composition} below.
\begin{theorem}\label{theorem: composition}
    The composition of $\ T$ iterations of updates in Algorithm \ref{alg: dpsgd with private selection} satisfies $T\cdot\delta_t$-approximated-$(\alpha, \sum_{i=1}^T(\mathbb{I}_i\cdot \epsilon_g^i(\alpha)+\epsilon_s^i(\alpha)))$-RDP and $(\sum_{i=1}^T(\mathbb{I}_i\cdot \epsilon_g^i(\alpha)+\epsilon_s^i(\alpha))+\frac{\log(1/\delta)}{\alpha-1}, T\cdot\delta_t+\delta)$-DP where $\mathbb{I}_i$ indicates whether the ``test'' in PTR is passed, $\alpha$ meets the constrain in Lemma \ref{lemma:dpsgd} and $\epsilon^i(\alpha)$ is the privacy budget of the $i_{th}$ iteration.
\end{theorem}


\section{Evaluation} \label{sec: evaluation}
In this section, we conduct experiments to answer three research questions:
{\bf RQ1}: Are the proposed attacks effective against FKGE?
{\bf RQ2}: Can \projectname mitigates these privacy threats?
{\bf RQ3}: How is the privacy-utility tradeoff?

\subsection{Experimental Setup}

\partitle{Datasets} We utilize two real-world knowledge graph benchmark datasets, FB15k-237 \cite{DBLP:conf/emnlp/ToutanovaCPPCG15} and NELL-995 \cite{DBLP:conf/emnlp/XiongHW17}. The details of these datasets are described in Appendix \ref{appendix: dataset}.

\partitle{Target FKGE Models}
We adopt four common KGE models: TransE \cite{DBLP:conf/nips/BordesUGWY13}, RotateE \cite{DBLP:conf/iclr/SunDNT19}, DisMult \cite{DBLP:journals/corr/YangYHGD14a} and ComplEx \cite{DBLP:conf/icml/TrouillonWRGB16}.
Their score functions are shown in Table \ref{table: score function}.
FKGE then collaboratively trains all four KGE models based on the following setting:
We randomly generate $1000$ triples, $500$ of them are included in the victim's training data as positive samples and the rest are left as negative samples.
The attack effectiveness is measured by whether the attacker can correctly distinguish the positive and negative samples.
The details of the hyper-parameters setting are shown in Appendix \ref{appendix: parameter}.


\partitle{Evaluation Metrics}
We use the harmonic mean of \textbf{\it Precision} and \textbf{\it Recall}, i.e., \textbf{\it F1-score}, to measure the attack success.
We report \textbf{\it Mean Reciprocal Rank (MRR)} and \textbf{\it Hits at N (Hits@N)} to evaluate the model performance, which follows the common practice in KGE literature.
Details about these metrics are shown in Appendix \ref{appendix: metric}.

\subsection{Attack Evaluation (RQ1)}
The overall performance of the proposed attacks is summarized in Table \ref{table: attack_f1score}.
The attacking rounds are set to every five rounds and we report the highest F1-score during the whole training procedure in the table.
CIP performs well on all four KGE algorithms with F1-Scrore around $0.8$.
CIA usually outperforms CIP on TransE and RotatE with F1-Scrore around $0.85$, while does not perform as well as CIP on DisMult and ComplEx.
The reason is that score functions of DistMult and ComplEx determine that triples are insensitive to reverse tail-entity embeddings.
SI only works on translational distance models, i.e., TransE and RotatE, because the relation embeddings in these models represent additive translations and can be well-clustered in the embedding space.
But the relation embeddings in bilinear models, i.e., DisMult and ComplEx, are multiplicative operations that can hardly be clustered.

In order to show how the performance of these attacks changes throughout the whole training procedure, we plot the F1-score of the attacks at different rounds in the left subfigure of Figure \ref{fig: f1&roc}.
We take the TransE and RotatE as examples because all three attack works on them.
We can see that in the early stage of the training, the F1-score of CIP and SI rises as the number of rounds increases.
In the late stages of the model training, the scores of these two attacks may no longer rise because the model gradually converges.
The F1-score of CIA peaks at the beginning of the training and decrease as the model converges since this attack focuses on the updates between two adjacent rounds. 
The F1-score varies with the thresholds, so we plot the ROC of the highest score of the corresponding attacks (on the left) in the right subfigure of Figure \ref{fig: f1&roc} to show the impact of thresholds.


\begin{figure}[htb]
	\centering
	\includegraphics[width=0.478\textwidth]{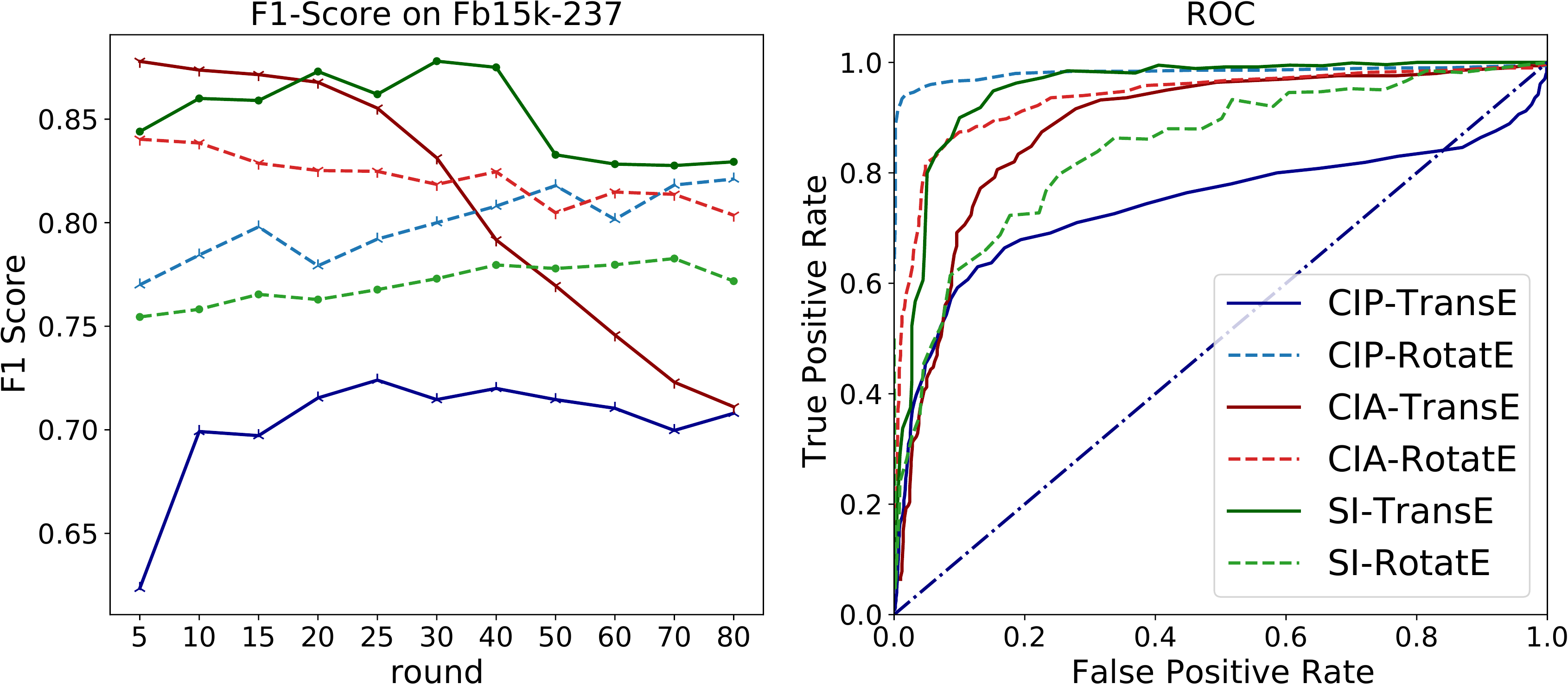}
	\vspace{-8mm}
	\caption{F1-Score of Attacks in Different Rounds and ROC.}
	\label{fig: f1&roc}
        \Description[The left figure is the F1-Score of attacks in different rounds and the right figure is the ROC of attacks.]{We can see that in the early stage of the training, the F1-score of CIP and SI rises as the number of rounds increases. In the late stages of the model training, the scores of these two attacks may no longer rise because the model gradually converges. The F1-score of CIA peaks at the beginning of the training and decrease as the model converges since the these attack focuses on the updates between two adjacent rounds. The F1-score varies with the thresholds, so we plot the ROC of the highest score of the corresponding attacks (on the left) in the right subfigure to show the impact of thresholds.}
		\vspace{-0.75em}
\end{figure}

\begin{figure*}[htb]
	\centering
	\includegraphics[width=1.0\textwidth]{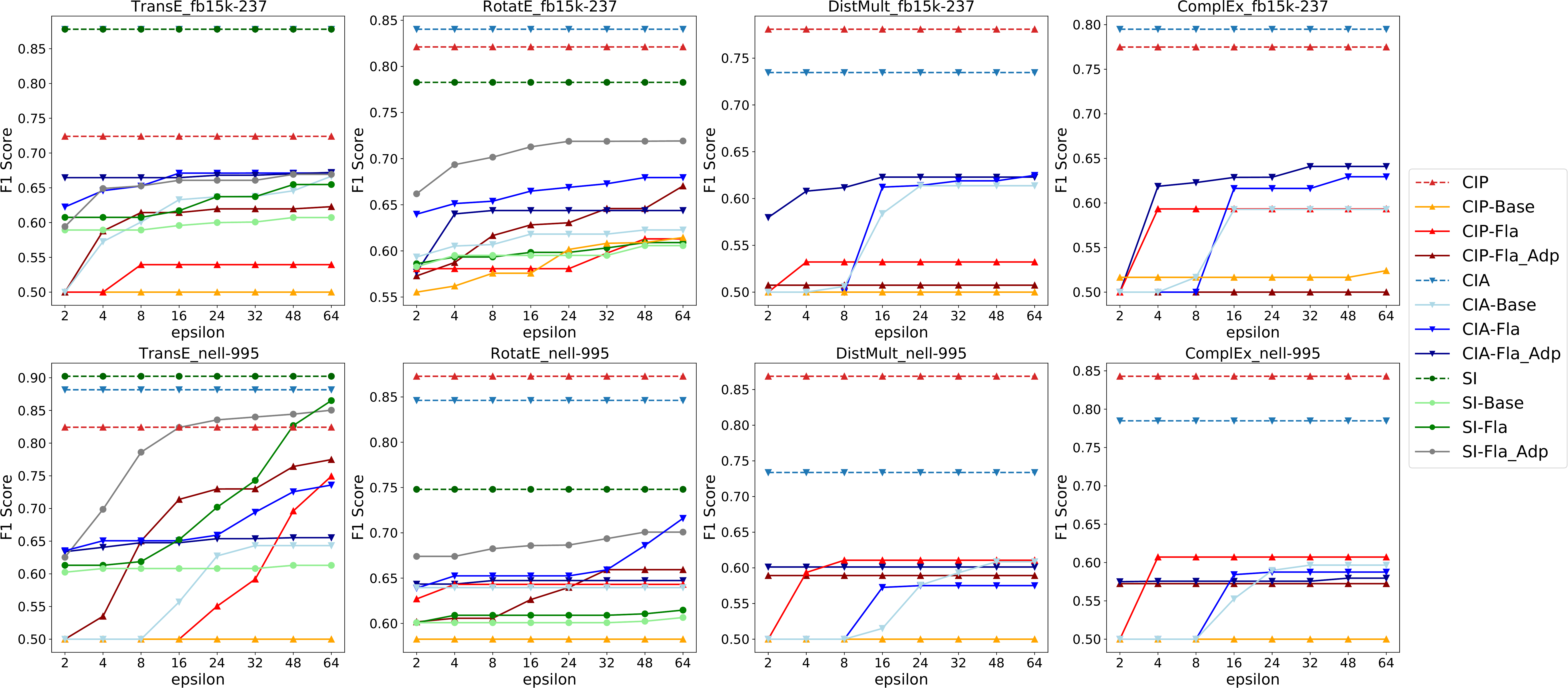}
	\vspace{-3ex}
	\setlength{\belowcaptionskip}{-2ex}
	\caption{F1-score of Attacks after Defense.}
        \Description[F1-score of three attack methods after defense.]{The baseline DPSGD for FKGE usually achieves the lowest F1-Score but at the price of intolerable utility decrease that will be shown later. On TransE and RotateE, DP-FLames-Adp  can better resist CIA than DP-FLames, while DP-FLames outperforms DP-FLames-Adp when resisting CIP and SI. On DistMult and ComplEx, DP-FLames-Adp and DP-FLames have comparable defense performance.}
	\label{fig: defense f1-score}
\end{figure*}

\begin{table*}[htb]
  \caption{\small Model Utility.}
  \footnotesize
  \vspace{-4mm}
  \begin{center}
   \setlength{\tabcolsep}{1.7mm}{
    \begin{tabular*}{\textwidth}{ @{\extracolsep{\fill}} lccccccccccccc}
     \toprule
     \multirow{2.5}{*}{Dataset} &\multirow{2.5}{*}{Setting} &\multicolumn{3}{c}{\textbf{TransE}}  & \multicolumn{3}{c}{\textbf{RotatE}} & \multicolumn{3}{c}{\textbf{DisMult}} & \multicolumn{3}{c}{\textbf{ComplEx}}\\
     \cmidrule(lr){3-5} \cmidrule(lr){6-8} \cmidrule(lr){9-11} \cmidrule(lr){12-14} 
     & & MRR & Hits@1 & Hits@10 & MRR & Hits@1 & Hits@10 & MRR & Hits@1 & Hits@10 & MRR & Hits@1 & Hits@10  \\
     \midrule
     \multirow{7}{*}{FB15K-237}
     &FKGE & 0.3606 & 0.2582 & 0.5720 & 0.3676 & 0.2671 & 0.5745 & 0.3138 & 0.2345 & 0.4718 & 0.3010 & 0.2277 & 0.4444 \\
     &Baseline ($\epsilon = 16$) & 0.1296 & 0.0492 & 0.2524 & 0.0506 & 0.0197 & 0.1180 & 0.0069 & 0.0060 & 0.0074 & 0.0067 & 0.0058 & 0.0068 \\
     &DP-Flame ($\epsilon = 16$) & 0.1843 & 0.0719 & 0.3949 & 0.0901 & 0.0477 & 0.1800 & 0.2342 & 0.1750 & 0.3533 & 0.2638 & 0.1999 & 0.3974 \\
     &DP-Flame-Adp ($\epsilon = 16$) & 0.3193 & 0.2233 & 0.5107 & 0.2946 & 0.2362 & 0.4548 & 0.2764 & 0.2074 & 0.4136 & 0.2874 & 0.2136 & 0.4285 \\
     \cmidrule(lr){2-14}
     &Baseline ($\epsilon = 32$) & 0.1248 & 0.0430 & 0.2575 & 0.0371 & 0.0138 & 0.0830 & 0.0084 & 0.0068 & 0.0098 & 0.0078 & 0.0066 & 0.0087  \\
     &DP-Flame ($\epsilon = 32$) & 0.2964 & 0.1962 & 0.4942 & 0.2740 & 0.1968 & 0.4228 & 0.2881 & 0.2137 & 0.4438 & 0.2784 & 0.2084 & 0.4226  \\
     &DP-Flame-Adp ($\epsilon = 32$) & 0.3239 & 0.2285 & 0.5143 & 0.3044 & 0.2241 & 0.4660 & 0.2728 & 0.2046 & 0.4072 & 0.2850 & 0.2114 & 0.4338  \\
     \midrule
     \multirow{7}{*}{NELL-955}
     &FKGE & 0.6866 & 0.6230 & 0.7976 & 0.7299 & 0.6824 & 0.7993 & 0.4120 & 0.3508 & 0.5183 & 0.3685 & 0.3072 & 0.4852 \\
     &Baseline ($\epsilon = 16$) & 0.1516 & 0.0035 & 0.3455 & 0.0265 & 0.0052 & 0.0576 & 0.0001 & 0 & 0 & 0.0001 & 0 & 0 \\
     &DP-Flame ($\epsilon = 16$) & 0.1419 & 0 & 0.3473 & 0.0151 & 0.0052 & 0.0209 & 0.1975 & 0.1396 & 0.2949 & 0.3340 & 0.3002 & 0.3962 \\
     &DP-Flame-Adp ($\epsilon = 16$) & 0.4687 & 0.2723 & 0.7330 & 0.3822 & 0.3438 & 0.4538 & 0.3745 & 0.2984 & 0.5166 & 0.3620 & 0.2932 & 0.4921 \\
     \cmidrule(lr){2-14}
     &Baseline ($\epsilon = 32$) & 0.1076 & 0 & 0.2618 & 0.0169 & 0.0017 & 0.0297 & 0.0002 & 0 & 0 & 0.0003 & 0 & 0  \\
     &DP-Flame ($\epsilon = 32$) & 0.3589 & 0.1658 & 0.6038 & 0.1099 & 0.0803 & 0.1623 & 0.3786 & 0.3054 & 0.4904 & 0.3747 & 0.2967 & 0.5061  \\
     &DP-Flame-Adp ($\epsilon = 32$) & 0.4814 & 0.2757 & 0.7365 & 0.5031 & 0.4555 & 0.6021 & 0.3912 & 0.3141 & 0.5201 & 0.3569 & 0.2740 & 0.4817  \\
     
     \bottomrule
   \end{tabular*}}
   \end{center}
   \label{tab: utility}
   \Description[Model Utility.]{The baseline DPSGD for FKGE has sharp utility drop on almost all FKGE models. DP-FLames shows higher utility than baseline, but still not good enough. DP-FLames-Adp achieves the best utility among the DP models, where under most of cases the utility even gets close to the non-DP models except the case of RotatE model on NELL-955.}
   \vspace{-1em}
\end{table*}

\subsection{Defense Evaluation (RQ2)}
We then study the effectiveness of \projectname in Algorithm \ref{alg: dpsgd with private selection} and show the results in Figure \ref{fig: defense f1-score}.
The solid lines denote the attack scores with defense at different $\epsilon$, while the dotted lines represent the scores without defense.
The DP methods with different $\epsilon$ start with the same initial noise magnitude $\sigma$, but end at different rounds when the $\epsilon$ runs out.
Our defense mechanisms can effectively defend against the attacks by significantly decreasing the F1-score of the attacks.
The baseline DPSGD for FKGE usually achieves the lowest F1-Score but at the price of intolerable utility decrease that will be shown later.
On TransE and RotateE, \projectnameadp can better resist CIA than \projectname, while \projectname outperforms \projectnameadp when resisting CIP and SI.
On DistMult and ComplEx, \projectnameadp and \projectname have comparable defense performance.

\begin{figure}[htb]
	\centering
	\includegraphics[width=0.478\textwidth]{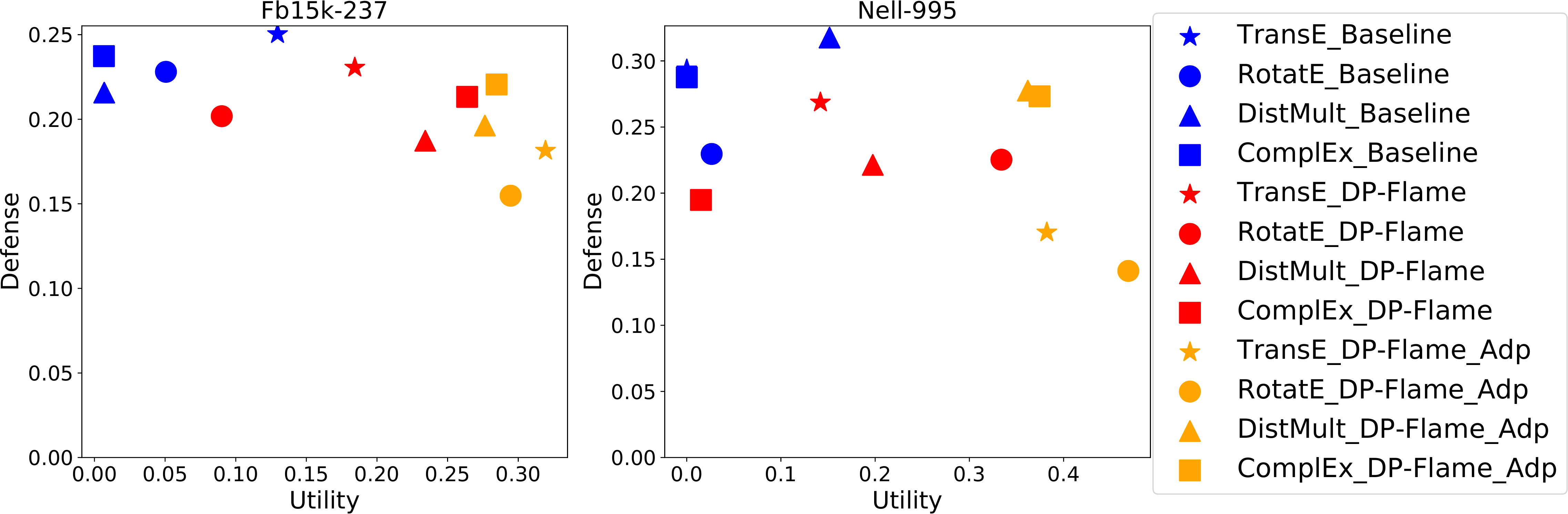}
	\vspace{-2em}
	\caption{Utility-Defense Trade-off.}
	\label{fig: utility&defense}
        \Description[The trade-off between utility and defense]{The defense capability is measured by the decrease in the averaged F1-Score of three attacks before and after defense, while the utility is measured with MRR. Both DP-FLames and DP-FLames-Adp obtains better utility-defense tradeoff.  DP-FLames-Adp always outperforms DP-FLames on utility, while DP-FLames shows better defense effect on TransE and RotatE.}
\end{figure}

\subsection{Model Utility Evaluation (RQ3)}
To study the adverse effect of DP defense on the model utility, we compare the DP models with non-DP models measured by  MRR, Hits@1, and Hits@10 as reported in Table \ref{tab: utility}.
Privacy budget are set to  $(16, 10^{-5})$ or $(32, 10^{-5})$, which have shown enough defense capability. 
The baseline DPSGD for FKGE has a sharp utility drop on almost all FKGE models.
\projectname shows higher utility than baseline, but still not good enough.
\projectnameadp achieves the best utility among the DP models, where in most cases the utility even gets close to the non-DP models except in the case of RotatE model on NELL-955.

In addition, we demonstrate the utility-defense tradeoff of the three defense mechanisms in Figure \ref{fig: utility&defense}.
The defense capability is measured by the decrease in the averaged F1-Score of three attacks before and after the defense, while the utility is measured with MRR.
Both \projectname and \projectnameadp obtain a better utility-defense tradeoff. 
\projectnameadp always outperforms \projectname on utility, while \projectname shows a better defense effect on TransE and RotatE.

\section{Conclusion} \label{sec: conclusion}
In this paper, we have comprehensively investigated the privacy risks of federated knowledge graph embedding (FKGE) from both attack and defense perspectives.
On the attack aspect, we have proposed three knowledge graph(KG) triple inference attacks on FKGE to expose its significant privacy vulnerability.
On the defense aspect, \projectname is proposed to provide rigorous differential privacy protection for FKGE, which exploits the sparse gradient property of FKGE by designing the private active gradient selection strategy.
An adaptive privacy budget allocation policy is further incorporated to dynamically adjust defense magnitude against the unbalanced privacy risks throughout the training procedure.
The experiment results demonstrate that the proposed defense can effectively defend against inference attacks with a modest utility decrease.

\clearpage
\balance
\begin{acks}
This work was supported by the National Key Research and Development Program of China (2021YFB3100300), National Natural Science Foundation of China (U20A20178 and 62072395), and CCF- AFSG Research Fund.
\end{acks}

\bibliographystyle{ACM-Reference-Format}
\bibliography{main}

\clearpage
\balance
\appendix
\section{Score Functions} \label{appendix: score function}

Score functions $f_r(h,t)$ of four typical KGE methods are described in Table \ref{table: score function}.
Note that ${\bf h}, {\bf r}, {\bf t}$ are embeddings of $h,r,t$, $\circ$ denotes Hadamard product and ${\rm Re}(\cdot)$ denotes the real vector component of a complex vector.

\begin{table}[h]
    \vspace{-2ex}
    \caption{Score functions $f_r(h,t)$ of typical KGE methods}
    \label{table: score function}
    \begin{tabular}{ccc}
      \toprule
      Model & Space & Score function $f_r(h,t)$  \\
      \midrule
      TransE & $h,r,t\in \mathbb{R}^d$ & $-\|{\bf h+r-t}\|$ \\
      RotatE & $h,r,t\in \mathbb{C}^d$ & $-\|{\bf h\circ r-t}\|$ \\
      DistMult & $h,r,t\in \mathbb{R}^d$ & ${\bf h}^\top{\rm diag}({\bf r}){\bf t}$ \\
      ComplEx & $h,r,t\in \mathbb{C}^d$ & ${\rm Re}({\bf h}^\top{\rm diag}({\bf r})\overline{{\bf t}})$ \\
      \bottomrule
    \end{tabular}
    \vspace{-2ex}
\end{table}

\section{Properties of RDP} \label{appendix: RDP}
One of the nice properties of RDP is the simplicity of composition.
\begin{lemma}\label{lemma: rdp composition}(\emph{Adaptive Composition of RDP})
If $\mathcal{M}_1(D)$ satisfies $(\alpha, \epsilon_1(\alpha))$-RDP, and $\mathcal{M}_2(D, \mathcal{M}_1(D))$ satisfies $(\alpha, \epsilon_2(\alpha))$-RDP, then the composition of $\mathcal{M}_1$ and $\mathcal{M}_2$ satisfies $(\alpha, \epsilon_1(\alpha)+\epsilon_2(\alpha))$-RDP.
\end{lemma}

RDP can be easily converted to DP.
\begin{lemma} \label{lemma: rdp conversion}(\emph{RDP Conversion to DP})
If a randomized algorithm $\mathcal{M}$ satisfies $(\alpha, \epsilon(\alpha))$-RDP, then $\mathcal{M}$ also satisfies $(\epsilon(\alpha)+\frac{\log(1/\delta)}{\alpha-1}, \delta)$-DP for any $\delta \in (0,1)$.
\end{lemma}


The approximated RDP also has the composition property.
\begin{lemma}\label{lemma: approximated rdp composition}(\emph{Adaptive Composition of approximated RDP})
If $\mathcal{M}_1(D)$ satisfies $\delta_1$-approximate-$(\alpha, \epsilon_1(\alpha))$-RDP, $\mathcal{M}_2(D, \mathcal{M}_1(D))$ satisfies $\delta_2$-approximate-$(\alpha, \epsilon_2(\alpha))$-RDP, then the composition of $\mathcal{M}_1$ and $\mathcal{M}_2$ satisfies $(\delta_1+\delta_2)$-approximate-$(\alpha, \epsilon_1(\alpha)+\epsilon_2(\alpha))$-RDP.
\end{lemma}

The approximated RDP can also be converted to DP.
\begin{lemma} \label{lemma:approximated rdp conversion}(\emph{Approximated RDP Conversion to DP})
If a randomized algorithm $\mathcal{M}$ satisfies $\delta_1$-$(\alpha, \epsilon(\alpha))$-RDP, then $\mathcal{M}$ also satisfies $(\epsilon(\alpha)+\frac{\log(1/\delta)}{\alpha-1}, \delta+\delta_1)$-DP for any $\delta \in (0,1)$.
\end{lemma}


\section{Related Work} \label{Appendix: related work}
\partitle{Membership Inference Attack}
Membership inference attack (MIA) aims to infer whether a data record has been used to train a machine learning model or not.
After Shokri et al.'s\cite{DBLP:conf/sp/ShokriSSS17} first MIA, extensive works have achieved successful attacks against various data domains, e.g., on image datasets \cite{DBLP:conf/nips/ZhuLH19, DBLP:conf/nips/GeipingBD020}, language datasets \cite{DBLP:journals/popets/HayesMDC19, DBLP:conf/ccs/SongR20}, or graph-structured datasets \cite{DBLP:conf/uss/HeJ0G021, DBLP:conf/sp/0011L0022, DBLP:journals/corr/abs-2110-02631}.
In particular, Nasr et al. \cite{DBLP:conf/sp/NasrSH19} and Hayes et al. \cite{DBLP:journals/popets/HayesMDC19} propose MIAs on federated learning models by obtaining gradients updates.
MIAs against centralized graph neural networks \cite{DBLP:conf/uss/HeJ0G021, DBLP:conf/sp/0011L0022} or knowledge graphs \cite{DBLP:journals/corr/abs-2104-08273} reveals that graph-structured data is also at risk of privacy leakage.

 
\vspace{-0.5em}

\partitle{Federated KGE and Privacy Protection}
FedE \cite{DBLP:conf/jist/ChenZYJC21} is proposed as the first federated KGE framework, which aggregates locally computed updates of entity embeddings to train a global model. 
FedR \cite{DBLP:journals/corr/abs-2203-09553} aggregate relation embeddings instead of entity embeddings.
Although the training data is preserved locally in these federated frameworks, the transferred gradient updates may potentially encode fine details of at least some of the training data \cite{DBLP:conf/ccs/AbadiCGMMT016, DBLP:journals/fgcs/MothukuriPPHDS21} and can be leveraged to compromise the privacy of the training data \cite{DBLP:conf/uss/CarliniTWJHLRBS21, DBLP:journals/tsc/TruexLGYW21, DBLP:conf/sp/NasrSH19}.
DP has been investigated for KGE privacy protection.
Han et al. \cite{DBLP:journals/ws/HanDGCB22} propose a DPSGD-based DP KGE algorithm for a single KG in the centralized learning. Peng et al. \cite{DBLP:conf/cikm/PengLSZ021} propose a PATE-GAN-based DP KGE algorithm for multiple KGs in the decentralized learning with the off-the-shelf PATE-GAN technique \cite{DBLP:conf/iclr/JordonYS19a}. However, both works are not rigorous enough in DP analysis, which makes their DP guarantees hold only for part of the algorithms as detailed in Appendix \ref{appendix: inaccurate DP}.

\section{Notation Table} \label{appendix: notation table}
\begin{table}[htbp] \small
    \caption{Important Notations.}
    \vspace{-2ex}
    \label{table: notation}
    \begin{tabular}{cc}
      \toprule
      Variable & Description\\
      \midrule
      $\mathcal{G}, \mathcal{E}, \mathcal{R}, \mathcal{T}$ & KG and the sets of entities, relations and triples\\
      $h,r,t$ & head entity, tail entity and relation\\
      ${\mE}, {\mR}$ & representing embedding matrix of entities and relations\\
      $\vh, \vt, \vr$ & representing embedding vectors of $h,r,t$\\
      $\mathcal{L}(\vh,\vr,\vt)$ & loss function of an embedding model\\
      $f_r(\vh,\vt)$ & score function that measure the plausibility of $(\vh,\vr,\vt)$\\
      $B, N, \lambda$ & batch size, number of triples in training set, learning rate\\
      $q$ & sampling probability\\
      $\mG$ & gradient matrix\\
      $C_1, C_2$ & clipping thresholds in DPSGD and private selection\\
      $\sigma, \eta, \Delta$ & noise allocation parameters\\
      $\epsilon_r$, $\sigma_p$, $\delta_t$ &private selection parameters \\
      \bottomrule
    \end{tabular}
    \vspace{-2em}
\end{table}



\section{Details of Adaptive Privacy Budget Allocation}\label{appendix: adaptive allocation}
In our experiment of active attack, the success rate peaks in the first few rounds of training and drops as the model converges. 
This motivates us to provide greater defense in the first few rounds and gradually lower the defense in the later rounds for better utility.

We propose to adaptively adjust the noise scale $\sigma$ based on the validation accuracy during the training procedure (Line 23-25 in Algorithm \ref{alg: dpsgd with private selection}).
Specifically, we set the initial $\sigma$ relatively large to defend against the client-initiate active inference attack and calculate the validation accuracy every few rounds.
A large $\sigma$ may prevent the model from converging as the learning approaches the local optimum, which is reflected in the fluctuated verification accuracy.
If the verification accuracy no longer increases, we decrease $\sigma$ by multiplying it with a coefficient $\eta<1$: 
\begin{small}
\begin{equation} \nonumber
    \sigma = 
    \begin{cases}
    \eta \cdot \sigma, &{\rm if}\ {\rm MRR}_t - {\rm MRR}_{t-1} < \Delta, \\
    \sigma, &{\rm if}\ {\rm MRR}_t - {\rm MRR}_{t-1} \ge \Delta,
    \end{cases}
\end{equation}
\end{small}
where ${\rm MRR}$ is a metric for knowledge graph completion and $\Delta$ is a threshold.
The validation can be performed on a public dataset, otherwise, the non-private validation may rise privacy concerns.
If no public dataset is available, we utilize predefined noise decay to periodically reduce the noise in an empirical manner.

\section{Theorem Proof} \label{appendix: proof}
\begin{theorem}
     The private selection in Algorithm \ref{alg: private selection} with sampling ratio $q$ satisfies $\delta_t$-approximated-$(\alpha, \epsilon_s(\alpha))$-RDP, where
    \begin{small}
    \begin{align}
        \nonumber &\epsilon_s(\alpha)=\frac{1}{\alpha-1}\log(p^2\tbinom{\alpha}{2}\min\{4(e^{\epsilon(2)}-1), e^{\epsilon(2)}\min\{2, (e^{\epsilon(\infty)}-1)^2\}\}\\
        \nonumber &+\sum_{j=3}^\alpha q^j\tbinom{\alpha}{j}e^{(j-1)\epsilon(j)}\min\{2,(e^{\epsilon(\infty)}-1)^j\}+1)
        \text{and } \epsilon(\alpha)=\frac{\alpha}{8\sigma_r^2}+\frac{\alpha}{2\sigma_p^2}.
    \end{align}
    \end{small}
\end{theorem}

\begin{proof}
The exponential mechanism we use at the first stage is a Gumbel noise style implementation \cite{DBLP:conf/nips/DurfeeR19} which is analyzed with RDP via a ``Bounded Range'' property for tighter bound. 
This implementation of exponential mechanism satisfies $(\alpha, \frac{\alpha}{8\delta_r^2})-RDP$.

We adopt a PTR variant with Gaussian noise \cite{DBLP:conf/aistats/0005W22} for the second stage.
Since there is a probability bounded by $\delta_t$ that the test in PTR may fail, the algorithm satisfies $\delta_t$-approximate-$(\alpha, \frac{\alpha}{2\sigma_p^2})$-RDP.

Thus, the private selection without subsampling satisfies $\delta_t$-approximate-$(\alpha, (\frac{\alpha}{8\sigma_r^2}+\frac{\alpha}{2\sigma_p^2}))$-RDP by composition.

Since the input of Algorithm \ref{alg: private selection} is computed from a subsampled batch of the original dataset(Line 11 in Algorithm \ref{alg: dpsgd with private selection}), the privacy guarantee can be amplified through the privacy amplification theorem for RDP with subsampled mechanism \cite{DBLP:conf/aistats/WangBK19}.
\end{proof}

\section{Experiment Setting}
\subsection{Datasets} \label{appendix: dataset}
We utilize two real-world knowledge graph benchmark datasets:
\begin{itemize}
    \item FB15k-237 \cite{DBLP:conf/emnlp/ToutanovaCPPCG15}: It is a subset with $237$ relations, $14505$ entities, and 310116 triples extracted from a large-scale knowledge graph Freebase \cite{DBLP:conf/sigmod/BollackerEPST08}.
    \item NELL-995 \cite{DBLP:conf/emnlp/XiongHW17}: It is a subset of NELL \cite{DBLP:conf/aaai/CarlsonBKSHM10} with $14505$ entities, $237$ relations, and 154213 triples which is built from the Web via an agent called Never-Ending Language Learner.
\end{itemize}
For each client, we randomly select a subset of entities and assign triples and relations according to the subset of entities.
We remove some triples to ensure that the triple sets corresponding to overlapping entities are different among clients.
And we use the same method to split validation and testing triples from the source dataset.


\subsection{Hyper-parameters Setting} \label{appendix: parameter}
For the federated KGE training, we set the batch size as $64$, and embedding dimension as $128$ for both entity and relation.
Adam \cite{DBLP:journals/corr/KingmaB14} with the initial learning rate of $0.001$ is used for SGD.
We set the number of negative sampling as $256$, the self-adversarial negative sampling temperature $\alpha$ as $1$, and the margin $\gamma$ as $10$.

When training the KGE model with differential privacy, we set the $\sigma, \sigma_r, \sigma_p$ as $1$, and the overall $\delta$ as $10^{-5}$. The clipping bound $C_1$ and $C_2$ is $1.2$ and $0.8$ respectively.
The sampling ratio $q$ can be computed with batch size $B$, and train data size $N$, e.g. $q=B/N$.
For the adaptive allocation approach, the decay coefficient $\eta$ is $0.95$, and the threshold $\Delta$ is $0.001$.
The overall privacy loss $(\epsilon, \delta)$ can be tracked with the above parameters during the training procedure.
For a preset $\epsilon$, we abort the training when the overall privacy loss reaches $\epsilon$.
We vary the preset $\epsilon$ in $\{2, 4, 8, 16, 24, 32, 48, 64\}$.

\subsection{Metrics} \label{appendix: metric}
\partitle{Attack} We use \textbf{\it F1-score} to measure the success rate of the attacks,
\[F_1=2\times \frac{Precision\times Recall}{Precision+ Recall}.\]

\partitle{Utility} The utility of a KGE model depends on its performance on the link prediction task, which focuses on predicting the missing part in a triple like $(h, r, ?)$.
The prediction procedure will report a ranked list of all possible entities, where the entity with a higher rank is more likely to be the right answer.
We report \textbf{\it Mean Reciprocal Rank (MRR)} and \textbf{\it Hits at N (Hits@N)} to evaluate the model performance on link prediction.
\[MRR=\frac{1}{|T|}\sum_{i=1}^{|T|}\frac{1}{rank_i},\]
\[Hits@N=\frac{1}{|T|}\sum_{i=1}^{|T|}\mathbb{I}(rank_i\le n),\]
where $T$ is the test triple set and $rank_i$ is the rank of the true answer.






\section{Supplementary Discussions}
\subsection{Attack Limitations}\label{appendix: attack limitations}
We would like to point out that there are still several limitations of the above attacks. Some of the limitations can be broken through when the server collaborates with a client. For example, the $\mathcal{A}_{SI}$ no longer requires the auxiliary dataset due to the relation information from the collusive client, while the $\mathcal{A}_{CIA}$ and $\mathcal{A}_{CIP}$ are no longer limited to overlapping entities since the collusive server could provide the complete entity embeddings collected from the victim.
However, some limitations still remain a problem. For instance, the $\mathcal{A}_{SI}$ has simpler entity embedding calculation only for transnational distance type of KGE models.
That being said, the proposed attacks suffice our purpose to reveal that FKGE suffers from significant privacy threats that should be defended against.


\subsection{DP Granularity} \label{appendix: granularity}
The two existing neighboring definitions for graph data \cite{DBLP:conf/icdm/HayLMJ09}, i.e., node-neighbor and edge-neighbor, are not suitable for the KGE model. 
The former defines two neighboring graphs if they differ in a single node. However, the KGE model outputs an embedding vector for each node (a.k.a., each entity in KG), which renders the node-neighbor definition inapplicable since the output embedding already reveals the existence of the corresponding entity. The latter regards neighboring graphs as differing in a single edge, which only reflects the existence of KG relations but fails to capture the type information in KG relations. 

\subsection{Inaccurate DP analysis of existing works} \label{appendix: inaccurate DP}
Han et al. \cite{DBLP:journals/ws/HanDGCB22} propose a DPSGD-based differentially private KGE algorithm for a single knowledge graph in the centralized learning setting \cite{DBLP:conf/ccs/AbadiCGMMT016}. Although they also attempt to only perturb the active gradient elements binding to the sampled entities, they fail to conduct the selection in a differentially private manner. As a result, the overall DP guarantee is not satisfied as believed by the paper, because without DP perturbation the positions of the updated entities precisely tell which entities are in the training dataset and selected for update during iterations.

Peng et al. \cite{DBLP:conf/cikm/PengLSZ021} propose a PATE-GAN-based differentially private KGE algorithm for multiple knowledge graphs in the decentralized setting with the off-the-shelf PATE-GAN technique \cite{DBLP:conf/iclr/JordonYS19a}. 
Unlike traditional PATE-GAN usage where it feeds the generator with random inputs, \cite{DBLP:conf/cikm/PengLSZ021} feeds the generator with embeddings extracted from training KGs that should have been protected with DP. 

\end{document}